\pdfoutput=1
\newif\ifArxiv
\Arxivtrue
\newif\ifFull
\Fulltrue
\ifArxiv
\documentclass[11pt]{article}
\usepackage{fullpage}
\else
\documentclass[10pt,journal,letterpaper,compsoc]{IEEEtran}
\fi
\usepackage{graphicx}
\usepackage{cite}
\usepackage{times}
\usepackage{latexsym}
\usepackage{algorithmic}
\usepackage{url}
\makeatletter
\def\@begintheorem#1#2{\sl \trivlist \item[\hskip \labelsep{\bf #1\ #2:}]}
\def\@opargbegintheorem#1#2#3{\sl \trivlist
      \item[\hskip \labelsep{\bf #1\ #2\ #3:}]}
\makeatother

\newenvironment{proof}{\noindent{\bf Proof:}}{\hspace*{\fill}\rule{6pt}{6pt}\bigskip}
\newtheorem{theorem}{Theorem}

\begin{document}
\renewenvironment{proof}{\noindent{\bf Proof:}}{\hspace*{\fill}\rule{6pt}{6pt}\bigskip}

\title{Learning Character Strings via Mastermind Queries,
       with a Case Study Involving mtDNA}

\ifArxiv
\author{
Michael T. Goodrich
\thanks{Michael T.~Goodrich is with the Department
of Computer Science, Univesity of California, Irvine, CA 92697-3435.
\hfil\break
E-mail and web page: see
\texttt{http://www.ics.uci.edu/\char`\~goodrich}.%
}}
\fi

\date{}

\maketitle 

\begin{abstract}

We study the degree to which a character string, $Q$,
leaks details about itself any time it engages in comparison protocols
with a strings provided by a querier, Bob, even if those protocols are
cryptographically guaranteed to produce no additional information other
than the scores that assess the degree to which $Q$ matches strings
offered by Bob.  We show that such scenarios allow Bob to play variants
of the game of Mastermind with $Q$ so as to learn the complete identity
of $Q$.  We show that there are a number of efficient implementations
for Bob to employ in these Mastermind attacks, depending on knowledge
he has about the structure of $Q$, which show how quickly he can
determine $Q$.  Indeed, we show that Bob can discover $Q$ using a
number of rounds of test comparisons that is much smaller than the
length of $Q$, under reasonable assumptions regarding the types of scores
that are returned by the cryptographic protocols and whether he can use
knowledge about the distribution that $Q$ comes from.
We also provide
the results of a case study we performed on a database of
mitochondrial DNA, showing the vulnerability of existing real-world DNA
data to the Mastermind attack.

\end{abstract}

\section{Introduction}
\emph{Mastermind}~\cite{c-m-83,k-cmm-77} 
is a game played between two players---a \emph{codemaker} and a
\emph{codebreaker}---using colored pegs.
\ifArxiv
(See Figure~\ref{fig-mastermind}.)

\begin{figure}[hbt]
\begin{center}
\includegraphics[width=3in]{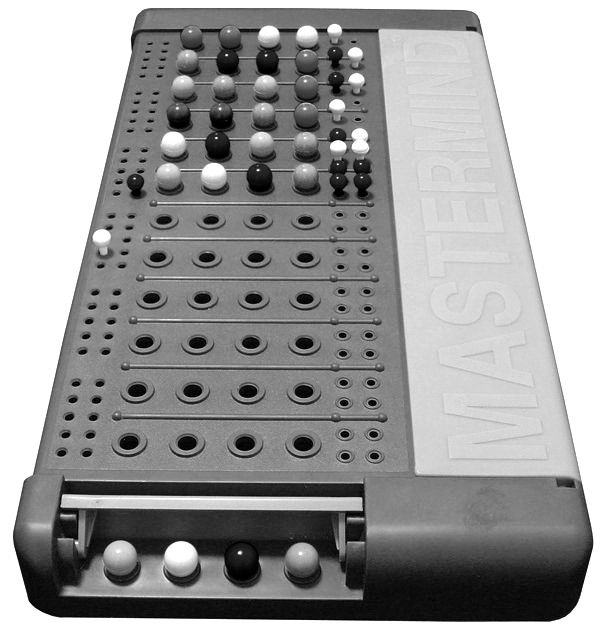}
\caption{The Mastermind game. The four large pegs in the middle are
used for guessing. The four smaller peg locations on the right are
used to score each guess---with black-peg and white-peg scores. And
the two pegs on the left are used to keep score across multiple games.
(This image is adapted from
{http://en.wikipedia.org/wiki/File:Mastermind.jpg}, by
User:ZeroOne, under the Creative Commons Attribution ShareAlike 2.0
License.)
}
\label{fig-mastermind}
\end{center}
\end{figure}
\fi

Viewed mathematically, Mastermind is abstracted as a game where the
codemaker selects a plaintext string\footnote{Throughout this paper we use
   the terms ``string,'' ``sequence,'' and ``vector'' synonymously.}, $Q$, 
of length $N$, whose elements are
selected from an alphabet of size $K$.
For consistency with the board game, the members of this alphabet are often
referred to as ``colors.''
The codemaker and codebreaker both know the values of $N$ and $K$,
and play consists of the codebreaker repeatedly making guesses, 
$V_1,V_2,\ldots$, about the identity of $Q$.
For each guess, $V_i$ the codemaker provides a score on how well
$V_i$ matches $Q$.
In \emph{double-count} Mastermind, which is the standard version based on the
board game, this score consists of a pair of two numbers:
\begin{itemize}
\item
A \emph{black} count, $b(Q,V_i)$, which is the number of elements in 
$V_i$ and $Q$
that match in both value and location. That is, 
\[
b(Q,V_i)=|\{j\colon\, V_i[j]=Q[j]\}|.
\]
\item
A \emph{white} count, $w(Q,V_i)$, which is the number of 
elements in $V_i$ that appear in $Q$ but in different locations than their
locations in $V_i$.
That is, letting $\pi$ denote an arbitrary permutation,
\[
w(Q,V_i) = \max_\pi|\{j\colon\, V_i[\pi(j)]=Q[j]\}|\,-\,b(V_i).
\]
\end{itemize}
In \emph{single-count} Mastermind, which has been less studied, the
codebreaker is given only the black count, $b(Q,V_i)$, for each guess, $V_i$.
(Note that it is impossible to solve the problem given only white-count
scores.)
The goal is for the codebreaker to discover $Q$ using a small a number of
guesses.

\subsection{Previous Related Work}
The original Mastermind game was
invented in 1970 by Meirowitz, as a board game having holes for 
vectors of length $N=4$ and $K=6$ colored pegs.
Knuth~\cite{k-cmm-77} subsequently showed 
that this instance of the Mastermind game can be solved in five guesses or less.
Chv{\'a}tal~\cite{c-m-83} studied the combinatorics of 
general Mastermind, showing that it
can be solved in polynomial time, in the $K\ge N$ case, 
using $2N \lceil\log K\rceil + 4N$ guesses, and
Chen {\it et al.}~\cite{cch-fhcaq-96} showed how this bound can be
improved, in this case, to 
$2N \lceil\log N\rceil + 2N + \lceil K/N\rceil + 2$ guesses.
Stuckman and Zhang~\cite{sz-mnc-05} showed that 
is NP-complete to determine if a sequence of guesses and responses in
general double-count Mastermind is satisfiable.
Goodrich~\cite{g-oacmg-09} shows that single-count (black-peg)
Mastermind satisfiability is NP-complete and that a specific vector
$Q$ can be guessed using a single-count (black-peg) query vector
that is of length
$N\lceil \log K\rceil + \lceil(2-1/K)N\rceil + K$.

Several researchers have 
explored privacy-preserving data querying
methods that can be applied to
character strings (e.g., see~\cite{akd-spsc-03,da-smc-01,FNP04}).
In particular, Atallah {\it et al.}~\cite{akd-spsc-03} and
Atallah and Li~\cite{al-sosc-05}
studied
privacy-preserving protocols for edit-distance string comparisons, such as
in the longest common subsequence (LCS) 
problem~\cite{h-lsacm-75,ir-acvlc-08,uah-bclcs-76},
where each party learns the
score for the comparison, but neither learns the contents of the string of
the other party.
Such comparisons are common in DNA sequence alignment comparisons, for
example.
Troncoso-Pastoriza {\it et al.}~\cite{tkc-pperd-07}
described a privacy-preserving protocol for searching for a certain
regular-expression pattern in a DNA sequence.
In last-year's Oakland conference, Jha {\it et al.}~\cite{jks-tppgc-08} give
privacy-preserving protocols for computing edit distance 
similarity scores between two genomic sequences, improving
the privacy-preserving edit distance algorithm of
Szajda {\it et al.}~\cite{spol-tpdp-06}.
Single-count matching results between two strings can be done in a
privacy-preserving manner, as well, using privacy-preserving set
intersection, e.g., using the method of 
Freedman {\it et al.}~\cite{FNP04},
Vaidya and Clifton~\cite{vc-ssica-05} or
Sang and Shen~\cite{ss-ppsip-07,ss-ppsib-08}.
The string matching problem can also be done using privacy-preserving dot
product computations~\cite{ae-nepps-07} 
or even general multi-party computation protocols
(e.g., see~\cite{dfknt-uscfm-06,gmw-hpamg-87,y-psc-82})
or systems~\cite{bnp-fssmp-08}.
Jiang {\it et al.}~\cite{jmcs-sddli-08} study a secure mulitparty
method for comparing a genomic sequence against every sequence in a
genomic database, providing a score indicating the match strength
between the query sequence and each sequence in the database.

In terms of the framework of this paper, the closest previous work is that of
Du and Atallah~\cite{da-psrda-01}, who studied a privacy-preserving protocol
for querying a string $Q$ 
in a database of strings, $D$, where comparisons are based on approximate
matching (but not sequence-alignment).
Their protocols assume that the parties are honest-but-curious,
however, so that, for instance, the database owner 
cannot introduce fake strings in
his database whose intent is to discover the identity of the query string,
$Q$. The attack model we explore in this paper, on the other hand,
allows for ``cheating'' in the comparison protocol, so that $D$ can introduce
strings whose sole purpose is to help him discover the identity of $Q$.

\subsection{Attack Scenarios}
In this paper we study the 
\emph{Mastermind attack} on string data, which is a way
that a genomic querier, Bob, can ``play'' a type of
Mastermind game with an unknown string, $Q$--for which $Q$'s owner,
Alice, thinks that she is comparing with Bob in a privacy-preserving 
manner---but instead Bob is discovering the full identity of $Q$.

The attack scenario is that
Alice repeatedly participates in privacy-preserving comparisons of
$Q$ to iteratively compare $Q$ with strings provided by Bob.
All is learned from each comparison is the score
measuring the similarity of the two strings ($Q$ and a string $V_i$
provided by Bob), with the score for each string
comparison being revealed to Bob (and possibly also Alice)
before the next comparison begins.
Bob's goal is to learn the complete identity of $Q$ with a reasonably small 
of comparisons.

We distinguish two versions of this attack scenario.
In the first scenario, the comparison between $Q$ and each string
$V_i$ provided by Bob is scored according to the single-count 
(black-peg) straight-match score,
\[
b(Q,V_i)=|\{j\colon\, V_i[j]=Q[j]\}|.
\]
In our second scenario, which is more common in genomic databases,
the comparison between $Q$ and each $V_i$ provided by Bob is scored
according to a sequence-alignment score,
\[
a(Q,V_i) = |\{(j,k)\in {\cal I}\colon\, V_i[j]=Q[k]\}|,
\]
where $\cal I$ is an ordered index set of pairs 
of integers so that if $(j,k)$ appears before $(l,m)$ in $\cal I$,
then $j<l$ and $k<m$.
This is also known as the \emph{longest common subsequence} 
(LCS)~\cite{h-lsacm-75,ir-acvlc-08,uah-bclcs-76} score between $Q$ and $V_i$.
(See Figure~\ref{fig:matches}.)
Incidentally,
as we observe below,
Levenshtein edit distance scores are strongly related 
to the LCS score, and our attack scenarios should be able to be
translated to this other measure, as well.

\begin{figure}[hbt]
\begin{center}
\includegraphics[width=3.3in]{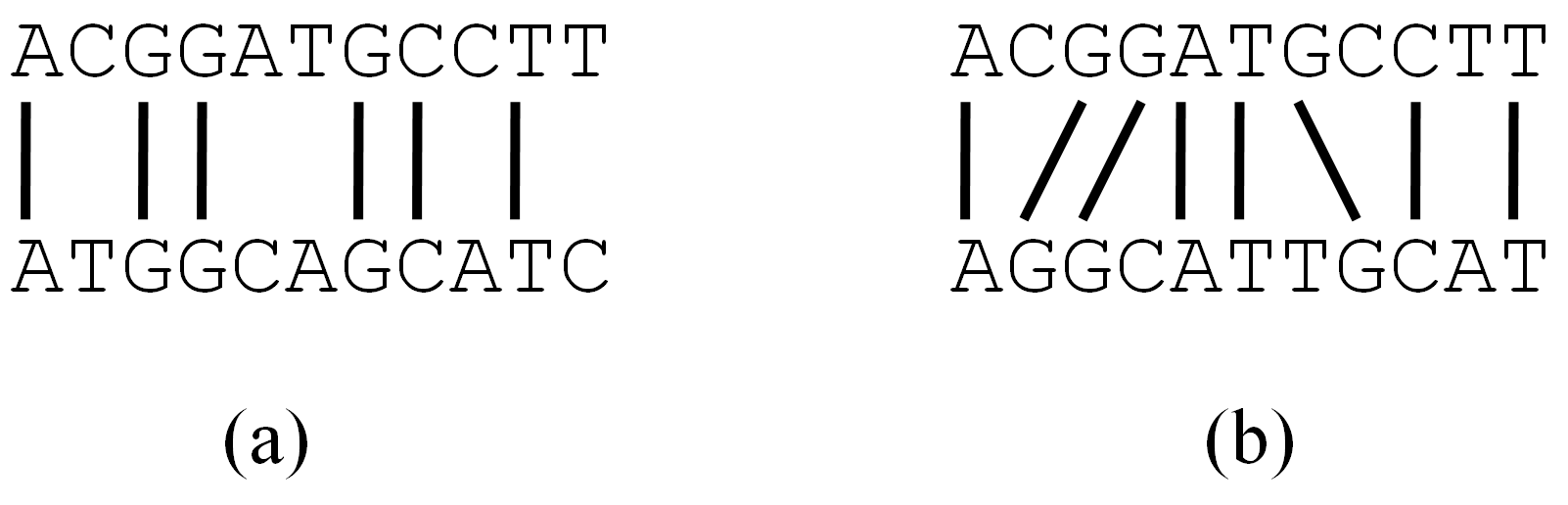}
\end{center}
\caption{Illustrating two types of matches between two DNA sequences.
(a) A single-count
(black-peg) straight-match. Note that the second
``A'' in the bottom string is not matched, since it doesn't line up
exactly with the second ``A'' in the top string.
(b) A sequence-alignment match. In going from the
top string to the botttom string, the
first ``C'' in the top string corresponds to a \emph{deletion} event,
the first ``C'' in the bottom string corresponds to an
\emph{insertion} event, and the penultimate characters in each string
correspond to a \emph{substitution} event.
}
\label{fig:matches}
\end{figure}

There are a number of motivating usage environments that could be
susceptible to Mastermind attacks.
For example, Bob could be a genomic database owner, storing 
genomic strings for a number of individuals, and Alice could be a
database user who is searching Bob's database to find the closest
match to a string $Q$ of interest.
Bob could, for instance, be the owner of a database of DNA from every
male attending a certain university and 
Alice could be an FBI agent searching through that database for a
match with DNA evidence gathered after a sexual assault.
Both parties in this example are likely to be under legal
restrictions not to reveal the complete identity of their strings
unless there is a match.
In another example, Alice could be the owner of a database
of genomic sequences and Bob could be an attacker trying to learn the
identity of a string $Q$ in Alice's database, e.g., which Bob can identify
only by an anonymized index, $j$.
In this case, Bob repeatedly does queries with each of his strings, $V_i$,
indexing into Alice's database using the name ``$j$'' 
to locate $Q$ and get Alice to do a privacy-comparison of $Q$ with $V_i$.
Bob could, for instance, be an employer trying to learn
the genomic sequence of a prospective employee, Charlie, by querying a
university DNA sequence database owned by Alice, which he could query
simply knowing the index of Charlie's DNA in Alice's database (e.g.,
Bob might be able to infer this index from Charlie's student
number).
In every case, Bob gets to ask Alice to compare her string, $Q$, to
each of his query strings, $V_i$, in a privacy-preserving manner
until these comparisons have leaked enough information that he can
easily infer the identity of $Q$.

\subsection{Our Results}
In this paper we study various aspects of the Mastermind attack, deriving the
following results.
\begin{itemize}
\item
We show that the problem of determining whether 
a sequence of Mastermind responses has a valid solution
is NP-complete even if each response is a sequence-alignment response.
\end{itemize}
At first, this might seem to provide some security for the privacy 
of the unknown string, $Q$, for it implies a degree of intractability to the
problem of learning a query
string $Q$ just from Mastermind responses involving $Q$.
Unfortunately, as was learned
with Knapsack cryptosystems~\cite{o-rfkc-90},
having the security
of a system be based on the difficulty of solving an NP-complete 
problem is no guarantee that it is safe in practice.
Indeed, such is the case for the security of genomic sequences being
susceptible to the Mastermind attack.
We show that character strings can be discovered by
surprisingly short sequence of guesses.
In particular, we also provide the following results:
\begin{itemize}
\item
We show that an arbitrary
query string, $Q$, of length $N$ from an alphabet of size $K$, can be
discovered with 
$(N+2)K$ queries, each of which reports the result of a 
sequence-alignment (LCS) test.
Such queries are common in genomic applications.
We also show that this bound can be further improved if the distribution of
characters in the alphabet follows Zipf's Law~\cite{n-plpdz-05}.
\item
We show how a Mastermind attacker can take advantage of
known distributional information for genomic data.
Armed with distributional knowledge about a query string, $Q$,
with respect to a reference string, $R$,
such as the Revised Cambridge Reference Sequence,
rCRS (GenBank accession number: AC 000021), 
the Mastermind attacker can discover $Q$ much quicker than in the
general cases, using either single-count or sequence-alignment responses.
\item
We provide experimental analysis of the distribution-based 
Mastermind attack for genomic data,
showing that, for a case study involving mitochondrial DNA (mtDNA),
either single-count responses or
sequence-alignment responses, the attack works surprisingly well. Given the
relative abundance of mtDNA data, and its ethnic sensitivity, we focus our
experiments on 
1000 human mtDNA sequences, showing that most can be discovered with
a Mastermind attack of just a few hundred guesses, even though mtDNA
sequences are typically over 16,500 bp long.
Given that current mtDNA databases
already have thousands of members
(e.g., see~\cite{brb-gppd-05}),
this experimental analysis shows that it
would be relatively easy for an attacker, Bob,
to interleave an undetected Mastermind attack
with privacy-preserving responses to actual sequences.
\end{itemize}

We conclude by discussing some of the issues that would have to be addressed in
order to defeat Mastermind attacks on genomic data, as well
as some possible directions for future research.

\section{Alternative Sequence Comparison Scores}
Throughout this paper, we assume that the attacker, Bob, can learn
the value of either a straight-match score, $b(Q,V_i)$, 
or a sequence-alignment score, $a(Q,V_i)$, 
between the unknown string, $Q$,
and each of his given strings, $V_i$.
These are not the only types of scores of interest with respect to
genomic data, however. 
So, before we discuss the privacy risks of genomic data 
from Mastermind attacks that use the $b$ or $a$ functions as scores,
let us discuss two 
other kinds of score functions and how they could alternatively
be used for similar attacks.

There are a number of score functions that measure 
the similarity between two strings.
We review two here, including how they can be reduced to 
similarity measures using the functions $b$ or $a$, for
comparing two strings, $Q$ and $V$.
\begin{itemize}
\item
\emph{Hamming distance:} the
Hammming distance, $H(Q,V)$, between $Q$ and $V$, is given by
\[
H(Q,V)=|\{j\colon\, V[j]\not=Q[j]\}|.
\]
That is, the two strings $Q$ and $V$ are
aligned in way that disallows insertions and deletions, and a score
is computed based on the number of substitutions needed to convert
$Q$ to $V$.
Note that, given a Hamming distance score, $H(Q,V)$,
we can compute a straight-match score as $b(Q,V) = |Q|-H(Q,V)$.
\item
\emph{Levenshtein distance:} the
Levenshtein distance, $L(Q,V)$, between $Q$ and $V$,
which is a kind of \emph{edit distance},
is the minimum number of insertions, deletions, and substitutions
needed to convert $Q$ into $V$ (or vice versa).
Note that, given a Levenshtein distance score, $L(Q,V)$, we can
compute a sequence-alignment score as
\[
a(Q,V) = \frac{|Q|+|V|-L(Q,V)}{2} .
\]
\end{itemize}
Thus, the Mastermind attacks we mention in this paper apply equally
well to systems that support string comparisons using Hamming
distance or Levenshtein distance.

\section{NP-Completeness of Sequence-Alignment Mastermind Satisfiability}
As mentioned above,
Stuckman and Zhang~\cite{sz-mnc-05} show that double-count Mastermind 
satisfiability is
NP-complete
and Goodrich~\cite{g-oacmg-09} shows that single-count (black-peg)
Mastermind satisfiability is also NP-complete (which applies equally
well for Hamming distance).

In the Sequence-Alignment Mastermind Satisfiability problem, we are given a
collection of Mastermind queries, $V_1,V_2,\ldots,V_N$, and the 
responses, $a(Q,V_1),a(Q,V_2),\ldots,a(Q,V_N)$, each of which is said to report the
sequence-alignment (LCS) score between each $V_i$ and an unknown vector, $Q$.
We are asked to determine if there indeed exists a vector 
$Q$ that satisfies all of these responses.

\begin{theorem}
\label{thm:np}
Sequence-Alignment Mastermind Satisfiability 
is NP-complete.
\end{theorem}
\begin{proof}
Our proof is an adaptation of the NP-completeness proof of 
Goodrich~\cite{g-oacmg-09} showing that single-count (black-peg)
Mastermind Satisfiability is NP-complete.
It is easy to see that Sequence-Alignment Mastermind Satisfiability is in NP.
For example, we could nondeterministically guess a vector $Q$ and then test
in polynomial time whether it satisfies all the responses,
$a(Q,V_1),a(Q,V_2),\ldots,a(Q,V_N)$.

To prove that Sequence-Alignment Mastermind Satisfiability is NP-hard,
we provide a reduction from 3-Dimensional Matching (3DM),
which is a well-known NP-complete problem
(e.g., see~\cite{gj-cigtn-79}).
In the 3DM problem, we are given three sets, 
$X=\{x_1,\ldots,x_n\}$, $Y=\{y_1,\ldots,y_n\}$, 
and $Z=\{z_1,\ldots,z_n\}$, of $n$
elements each.
In addition, we are given a set $T$ of $m$ triples, 
$\{(x_{i_1},y_{j_1},z_{k_1}),\ldots,
(x_{i_m},y_{j_m},z_{k_m})\}$, whose elements are respectively taken from the
three sets, $X$, $Y$, and $Z$.
The problem is to determine if there is a subset of triples such that each
element in $X$, $Y$, and $Z$ appears in exactly one triple in 
this subset.

Suppose, then, that we are given an instance of the 3DM problem, as described
above.
We consider the unknown vector, $Q$, to consist of the following vector of
variables:
\[
(X_1,\ldots,X_{2n};\,Y_1,\ldots,Y_{2n};\,Z_1,\ldots,Z_{2n};
\,T_1,\ldots,T_{2m-1}),
\]
where the semi-colons are used for the sake of notation to separate the four
sections in the unknown vector, $Q$.
We perform our reduction by constructing a sequence of guess vectors,
$V_0,V_1,\ldots,V_N$, together with their sequence-alignment responses, 
$a(Q,V_0),a(Q,V_1),\ldots,a(Q,V_N)$,
so that there is a satisfying vector $Q$ for these responses if and only if
there is a solution to the given instance of the 3DM problem.

Our construction begins by setting the number of colors, $K$, to be $m+2$.
Intuitively, there is a color associated with each triple in $T$, 
plus a ``null'' color, $\phi$, which is guaranteed to appear nowhere in our
unknown vector, $Q$, and a separator color, $\mu$, which occurs in every
other (even-indexed) position of $Q$.
We begin our sequence of queries
with four special ``enforcer'' queries.
The first two of these are
\[
V_0 = (\phi,\ldots,\phi;\,
\phi,\ldots,\phi;\,
\phi,\ldots,\phi;\,
\phi,\ldots,\phi),
\]
which has response $a(Q,V_0)=0$, and
\[
V_1 = (\mu,\ldots,\mu;\,
\mu,\ldots,\mu;\,
\mu,\ldots,\mu;\,
\mu,\ldots,\mu),
\]
which has response $a(Q,V_1)=3n+m-1$.
Intuitively, $V_0$ enforces the fact that the null color, $\phi$, 
appears nowhere in the unknown vector,
and $V_1$ enforces the fact that the separator color, $\mu$, appears exactly 
often enough to separate every other (non-$\mu$) character in the 
unknown vector.
So as to better understand the characteristics of the other queries, let us
set $h=3n+m-1$, the number of $\mu$ colors in our unknown vector $Q$.
We then define two additional enforcer queries,
\begin{eqnarray*}
V_2 &=& (\phi,\mu,\ldots,\mu,\phi,\mu;\,
\phi,\mu,\ldots,\phi,\mu; \\
&& \phi,\mu,\ldots,\phi,\mu;\,
1,\mu,1,\mu,\ldots,\mu,1),
\end{eqnarray*}
which has response $a(Q,V_2)=h+n$, and
\begin{eqnarray*}
V_3 &=& (\phi,\mu,\ldots,\mu,\phi\mu;\,
\phi,\mu,\ldots,\mu,\phi\mu; \\
&& \phi,\mu,\ldots,\mu,\phi\mu;\,
0,\mu,0,\mu,\ldots,\mu,0),
\end{eqnarray*}
which has response $a(Q,V_3)=h+m-n$.
Intuitively, 
$V_2$ enforces a counting rule that exactly $n$ of the $T_i$'s will be set to
$1$, and $V_3$ 
enforces a counting rule that the remaining $m-n$ of the $T_i$'s will be set to
$0$.
For each triple, $T_s=(x_{i_s},y_{j_s},z_{k_s})$, we
construct three query vectors, as follows.

\medskip
\noindent
$V_{s,1} = $ 

\smallskip
\noindent
$(\phi,\mu,\ldots,\mu,\phi,\mu,s,\mu,\phi,\mu,\ldots,\mu,\phi,\mu;\,
\phi,\mu,\ldots,\mu,\phi,\mu;$ \hfil\break
$\phi,\mu,\ldots,\mu,\phi,\mu;\,
\phi,\mu,\ldots,\mu,\phi,\mu,0,\mu,\phi,\mu,\ldots,\mu,\phi)$,

\bigskip
\noindent
where the $s$ is in position $2i_s-1$ in the first group and 
the $0$ is in position $2s-1$ in the fourth group.
This vector has response, $a(Q,V_{s,1})=h+1$.

\medskip
\noindent
$V_{s,2} = $ 

\smallskip
\noindent
$( \phi,\mu,\ldots,\mu,\phi,\mu;\,
\phi,\mu,\ldots,\mu,\phi,\mu,s,\mu,\phi,\mu,\ldots,\mu,\phi,\mu;$ \hfil\break
$\phi,\mu,\ldots,\mu,\phi,\mu;\,
\phi,\mu,\ldots,\mu,\phi,\mu,0,\mu,\phi,\mu,\ldots,\mu,\phi)$,

\bigskip
\noindent
where the $s$ is in position $2j_s-1$ in the second group and 
the $0$ is in position $2s-1$ in the fourth group.
This vector has response, $a(Q,V_{s,2})=h+1$.

\medskip
\noindent
$V_{s,3} = $ 

\smallskip
\noindent
$( \phi,\mu,\ldots,\mu,\phi,\mu;\,
\phi,\mu,\ldots,\mu,\phi,\mu;$ \hfil\break
$\phi,\mu,\ldots,\mu,\phi,\mu,s,\mu,\phi,\mu,\ldots,\mu,\phi,\mu;$ \hfil\break
$\phi,\mu,\ldots,\mu,\phi,\mu,0,\mu,\phi,\mu,\ldots,\mu,\phi)$,

\bigskip
\noindent
where the $s$ is in position $2k_s-1$ in the third group and 
the $0$ is in position $2s-1$ in the fourth group.
This vector has response, $a(Q,V_{s,3})=h+1$.
Intuitively, these three responses collectively form a ``chooser'' gadget,
where we will either have $T_{2s-1}=0$ or the three variables
$X_{2i_s-1}$,
$Y_{2j_s-1}$,
and
$Z_{2k_s-1}$,
will each be set to have color $s$ 
(and $T_{2s-1}=1$).
Moreover, note that there are $m$ odd-index positions in the $T$, and each of
them has to match either a $0$ or $1$ color.

This reduction can clearly be done in polynomial time. So all that remains 
is for us 
to show that it works.
Suppose, then, that there is a possible solution to the given
instance of 3DM. Then for each chosen triple, 
$T_s=(x_{i_s},y_{j_s},z_{k_s})$, we can assign colors
$T_{2s-1}=1$, $X_{2i_s-1}=s$, $Y_{2j_s-1}=s$, and 
$Z_{2k_s-1}=s$,
which will satisfy each of the $V_{s,1}$, $V_{s,2}$, and $V_{s,3}$ vector
responses for this value of $s$.
Likewise, setting $T_{2s-1}=0$ will 
satisfy each of the $V_{s,1}$, $V_{s,2}$, and $V_{s,3}$ vector
responses for a triple $T_{2s-1}$ that is not chosen.
Finally, given that there are $n$ chosen vectors, we will satisfy the four
preliminary vector responses as well.

Suppose, alternatively, that we have a vector $Q$ that satisfies all 
our vector responses.
We know that each $X_i$, $Y_j$, and $Z_k$ must be assigned a color other than
$\phi$. 
Moreover, every even-indexed position in $Q$ must be assigned the color $\mu$
and every odd-indexed position must be a color other than $\mu$, because
there are exactly $h=3n+m-1$ instances of $\mu$
in $Q$ and we have introduced a query that enforces the fact that there is
exactly one non-$\mu$ color between every consecutive pair of $\mu$-colored
positions.
Since there are only $m+2$ colors, this implies each 
odd-indexed position
$X_{2i-1}$, $Y_{2j-1}$, and $Z_{2k-1}$ 
must be assigned a color corresponding to a triple
number, $s$, that is, it is not assigned $\phi$ or $\mu$.
If the corresponding $T_{2s-1}=1$, then in order to have satisfied the vectors
$V_{s,1}$, $V_{s,2}$, and $V_{s,3}$,
we must have set
$X_{2i_s-1}=s$, $Y_{2j_s-1}=s$, and $Z_{2k_s-1}=s$,
which implies we can include the triple
$(X_{i_s},Y_{j_s}Z_{k_s})$ in our matching.
If $T_{2s-1}=0$, then we do not include this triple in our matching.
By the vector responses $V_2$ and $V_3$, we know that the number of 
triples chosen in this way is 
exactly $n$. Thus, we have found a valid 3-dimensional matching.
\end{proof}

Thus, it is extremely unlikely
that we will be able to find a polynomial-time algorithm that can
always satisfy arbitrary Mastermind sequence-alignment query strings, 
or even single-count queries~\cite{g-oacmg-09}.
Unfortunately, this is not the same as a 
guarantee of security for the kinds of query strings
that would result from an interaction between a Mastermind attacker,
Bob, and a character string owner, Alice, where Bob is trying to learn
Alice's string, $Q$, through a sequence of privacy-preserving string
comparisons. 
For we show, in the sections that follow, that such query
strings, $Q$, can be discovered fairly efficiently using the Mastermind attack.

\section{The Mastermind Attack for Sequence-Alignment Queries}
\label{sec:align}
Recall that in a
\emph{sequence-alignment} query
we wish to compare two strings $Q$ and $V$, where
the score for a match is the length of the longest common subsequence
(LCS)~\cite{h-lsacm-75,ir-acvlc-08,uah-bclcs-76} between $Q$ and $V$.
Several researchers have studied this problem and have come up with
privacy-preserving protocols to determine such scores
(e.g., see~\cite{akd-spsc-03}).
In this section, we show that performing such a series of 
sequence-alignment queries
with Bob is susceptible to a type
of Mastermind attack of its own.

Suppose we are given an unknown string $Q$ of length $N$ over an
alphabet of size $K$, the members of which we call ``colors.''
Suppose further that we are going to engage in a protocol with Bob
to test $Q$ against strings provided by Bob,
where each test returns the length of a longest common subsequence
between $Q$ and one of Bob's strings.
That is, we score matches using the sequence-alignment scoring
function, $a(Q,V)$, for a guess vector $V$, 
which is the length of
a longest common subsequence between $V$ and $Q$.
We are interested in this section on studying an efficient scheme for
Bob to discover $Q$ using this query scheme.

A Mastermind-attack algorithm for Bob begins as follows:
\begin{itemize}
\item
Bob begins by guessing $K$ vectors, $V_1,V_2,\ldots,V_K$, with 
each vector $V_i$ consisting of elements of all the same color, $i$.
\end{itemize}

The subsequence alignment score for each of the initial guesses 
will tell Bob the cardinality of each color in $Q$.
Let us now imagine that we reorder the colors so that they are listed
$1$ to $K$ in nondecreasing order of how often they each appear in $Q$.
Thus, color $1$ is now the least frequent color in $Q$ and $K$ is the
most frequent color.
Our algorithm continues by incrementally building up a vector $W$,
such that $W$ either completely matches all its characters with $Q$
(in the specified order) or it misses by just one character.
Initially, we set $W$ to be a vector consisting of exactly $c_1$ elements
of color $1$, so that if we were to guess $W$,
then we would get a score of $a(Q,W)=c_1$.
We allow indexing and insertion into $W$ so that 
we can add a character before the $i$th element in $W$ for $i=1$ to
$|W|+1$ (with an insertion ``before'' position $|W|+1$ taken to mean
an insertion just after position $|W|$, the last position in $W$).
Our algorithm for Bob's Mastermind attack continues shown in
Figure~\ref{fig:alg-mast}.

\begin{figure}[hbt]
\begin{algorithmic}[100]
\FOR[take each color in turn] {$k=2$ to $K$}
  \STATE
  Set $i=1$ \{position in $W$ where to insert items\}
  \STATE Set $j=0$ \{count of number of items of color $k$ found\}
  \WHILE[find the places for color $k$] {$j<c_k$}
    \STATE Add a color $k$ item just before the $i$th item in $W$.
    \STATE Make a guess for $W$ to learn the value of $a(Q,W)$.
    \IF[all of $W$ matches] {$a(Q,W)=|W|$}
      \STATE Increment $i$ and $j$.
    \ELSE[there's one too many of color $k$ before $i$]
      \STATE Remove the color $k$ item before $i$.
      \STATE Increment $i$.
    \ENDIF
  \ENDWHILE
\ENDFOR
\end{algorithmic}
\caption{The sequent-alignment learning algorithm.}
\label{fig:alg-mast}
\end{figure}

Note inductively that, at the end of each iteration of the 
the while-loop, every character in $W$ matches in $Q$, that is,
$a(Q,W)=|W|$.
Thus, any time the if-statement finds that
$a(Q,W)\not=|W|$, then we have just added an item of color $k$ in a
place where it cannot match any item without causing a previously-matched
neighboring item to mis-match what it previously could match.
Therefore, in each iteration of the for-loop,
the algorithm correctly finds all the places where items of color $k$
fit with respect to items of colors $1$ to $k-1$.
So, when the algorithm completes, we have $W=Q$; that is, we have
learned $Q$.

Consider now the analysis of this algorithm. 
Note that in each iteration of the while-loop, we increment $i$, our
index into $W$, and that at the end of the while loop the length of
$W$ is $c_1+c_2+\cdots+c_k$, where $k$ is the index of the for-loop.
Thus, the total number of queries made is at most
\[
K + \sum_{i=1}^{K}\ \sum_{j=1}^{i} c_j,
\]
which is the same as
\[
K + \sum_{i=1}^{K} (K-i+1)c_i,
\]
since each term $c_i$ appears $K-i+1$ times in the double sum.
Let us perform a substitution of variables, where we let
$d_1,d_2,\ldots, d_K$ denote the cardinalities of the colors in $Q$
in nonincreasing order, so $d_1$ is the most frequent color and $d_K$
is the least frequent.
Then we can rewrite the total number of queries performed to be bounded by
\[
K+ \sum_{i=1}^{K} id_{i}.
\]
Note that, by definition, $d_{i}\le N/i$, for otherwise, $d_{i}$
could not be the $i$th largest-cardinality color.
Thus, the total number of queries is at most
\begin{eqnarray*}
K+ \sum_{i=1}^{K} i(N/i) &=& K+KN \\ 
			      &=& (N+1)K.
\end{eqnarray*}
This is the number of tests done by Bob, the Mastermind attacker,
making no additional assumptions about the distribution of colors in
the query string, $Q$.

This analysis can be refined, however,
if the colors are distributed in $Q$ according to Zipf's
Law~\cite{n-plpdz-05}, which in this context would imply that 
\[
d_i \le \frac{N}{i^s H_{N,s}} ,
\]
where $H_{N,s}$ is the $N$-th Harmonic number of order $s$,
\[
 H_{N,s} = \sum_{i=1}^N 1/i^s,
\]
and $s$ is between $1$ and $2$, inclusive.
In this case, the total number of guesses done by Bob would be
at most
\begin{eqnarray*}
K + \sum_{i=1}^{K} \frac{iN}{i^sH_{N,s}} &\le& K+\frac{KN}{H_{N,s}},
\end{eqnarray*}
for $s\ge 1$.
Thus, we have the following:

\begin{theorem}
\label{thm:seq-align-q}
Given an unknown length-$N$ string $Q$, defined on an alphabet of
size $K$, a malicious Mastermind attacker can discover 
$Q$ in polynomial time using 
$(N+1)K$ sequence-alignment tests
tests against $Q$, each of which reveals only the length of a longest
common subsequence between
$Q$ and the test string match.
If the cardinalities of elements of $Q$ follow 
Zipf's Law, with parameter $s\ge 1$, 
then a malicious Mastermind attacker can discover $Q$
using at most $K+KN/H_{N,s}$ sequence-alignment tests.
\end{theorem}

\section{Exploiting Data Distributions}
\label{sec:dist}
Up to this point, we have focused on how the
Mastermind attacker, Bob, could learn a general string $Q$ using the
types of queries typically asked of genomic databases, even
if those queries are privacy preserving.
In this section, we explore how Bob can significantly improve the
effectiveness of the Mastermind attack if he exploits information,
which is publicly available, about the distributions of
the character strings of interest.
Moreover, to drive the point home, we provide a case study showing
the effectiveness of such Mastermind attacks on a real-world
genomic database, in the section that follows.

Genomic sequences typically have a great deal of similarity. Indeed, recent 
compression schemes have shown that it is effective to view a genomic sequence
with respect to a compression scheme 
that represents a sequence in terms of its
differences with a reference sequence, $R$
(e.g., see~\cite{baldicompression07}).
That is, we can start from a reference sequence, $R$, 
which contains the most common components of a typical genomic sequence.
Then we define each other sequence, $Q$, in terms of its differences with $R$.
Each difference is defined by an index location, $i$, in $R$ and an operation
to perform at that location, such as a substitution, insertion, or deletion.

This difference pattern is present,
for example, in human mitochondrial DNA, which is the type of genomic
data we use in our case study. 
This type of of DNA, which, as we have already mentioned,
is inherited only through the maternal
line and is already available in sequenced form in sizeable enough
quantities to support obfuscated Mastermind attacks.
Moreover, because it is passed only though the maternal line, it
functions as a highly tuned notion of race, allowing researchers in
some cases to trace a person's ancestry to individual villages.
Thus, mitochondrial DNA is highly sensitive from a privacy-protection
viewpoint.

As shown in recent work of Baldi {\it et al.}~\cite{baldicompression07},
mitochondrial DNA 
sequences can be encoded in significantly-compressed form by using a
standard reference sequence~\cite{brandon04,ruiz07}.
This reference sequence, $R=\mbox{rCRS}$, is 16,568 bp long.
So, in terms of the notation used above, we have $N=16568$ and
$K=4$, since there are 4 types of base pairs possible.
But these parameters suggest that there is more variation in the data
than actually occurs. 

In fact, the vulnerability of DNA sequences to the Mastermind attack is much
worse than this in practice.
For example, there are a limited number of
locations along the reference
sequence where any changes appear statistically in the mitochondrial DNA data. 
So let us use $M$ to
denote the number of different possible locations where any query sequence 
might differ from the reference sequence, $R$.
Worse yet, from a privacy-preservation standpoint, the average number
of difference between any human DNA sequence and the
reference is orders of magnitude smaller than $M$ in practice.
(We explore these statistics in detail below.)
Here we show how a Mastermind attack can exploit these statistical 
properties of genomic data.

\subsection{The Substitution-Only Case}
\label{sec:deterministic}

In this section, we explore the version of the Mastermind
attack where the attacker, Bob, 
engages in a series of privacy-preserving protocols with Alice,
each of which reveals only the single-count straight-match score 
between Alice's string, $Q$, and strings provided by Bob, 
in an iterative online fashion (recall Figure~\ref{fig:matches}a).
In the attack model we consider,
Bob is allowed to use self-constructed sequences in comparisons with $Q$,
from which he learns the value of $b(Q,V_i)$ for each of his query strings,
$V_i$.

Given additional knowledge of the distributional properties
of DNA data, we can construct a Mastermind attack to take this
knowledge into consideration.
In this case, we make the assumption that the unknown string, $Q$, differs from
a reference string $R$ only through a relatively small
number substitutions, which is true
for example, for 45\% of the mitochondrial DNA data.
(We will explore the more general case later in this section.)

Our algorithm is an adaptation of an algorithm of Goodrich~\cite{g-oacmg-09} 
for solving the boardgame version of Mastermind 
to the specific case of a Mastermind
attack on a string $Q$ relative to a reference string $R$.

We begin the attack for Bob
by having him perform a query against $Q$ with a reference sequence, $R$.
For any string, $Q$, let $s(Q)$ denote the number of
substitutional differences $Q$ has with the reference sequence, $R$.
Note, then, that our first query
(for the reference string $R$ itself) allows us to 
determine the value of $s(Q)$, using the formula
\[
s(Q)=N-b(Q,R).
\]
For example, $R$ could be a genomic sequence derived from a sequencing
of the DNA of a specific reference human or it could be a canonical genomic
reference sequence derived from analyzing commonalities among a number
of human sequences.
Even though few humans
have presently had their complete genomes 
sequenced~\cite{venter01,humangenome01,levy07},
any of these could serve as a reference, $R$, for a Mastermind attack
on a complete genome sequence.
For the more wide-spread instances of mitochondrial DNA,
the Revised Cambridge Reference Sequence
(rCRS) (GenBank accession number: AC 000021) is commonly used as a 
mtDNA reference sequence~\cite{brandon08,ruiz07,brandon04}, and it
could serve as the sequence $R$ in a Mastermind attack on a
mitochondrial DNA sequence. 

Imagine that we cyclically order the $K$ characters
in our alphabet, so, for instance, if our alphabet is \{A,C,G,T\},
then we could use the cyclic ordering (A,C,G,T,A,C,G,T,$\ldots$).
Note that this ordering allows us to choose any character as a base
color, i.e., a ``color $0$,'' and then specify all other characters as 
offsets from that base.
For example, in the DNA case, we could pick ``C'' as the base, color
$0$, in which case ``G'' becomes color $1$, ``T'' becomes color $2$,
and ``A'' becomes color $3$.
Or we could pick ``T'' as the base, color $0$, in which case ``A''
becomes color $1$, ``C'' becomes color $2$, and ``G'' becomes color
$3$.

In the context of a Mastermind attack,
we consider each character, $R_i$, in the reference sequence, $R$, to be
color ``0'' for that position, $i$. 
Viewed Mathematically, we can then number the $K-1$ remaining
characters, according to our cyclic ordering, as offsets from these
respective color $0$'s.
Assuming that Bob's first guess, of $R$, is not a perfect match for
the query sequence, $Q$, then we can view Bob's remaining task as that
of determining the cardinality and location of all the non-zero
offset values for positions in $R$.
In fact, if we think of the characters in the respective positions of $R$ as
the respective color $0$'s for those positions, then we can view the
remaining task as that of 
determining the locations of the colors $0$ through $K-1$.

After Bob makes his initial guess using $R$, we then
have him perform $K-1$ additional queries, each of which is a vector of
elements that are all the same offset from $R$, i.e., a vector of all the
same ``colors'' with respect to $R$, but only at the $M$ places that are
statistically possible locations for a substitution.
Thus, let us assume we can view $Q$ as now consisting of just the $M$ places
where substitutions may occur (for the other locations we simply repeat a
guess for color $0$ every time).
This allows us to initially know the cardinality, 
$c_0,c_1,\ldots,c_{K-1}$, of every
(offset) color in the (compressed) unknown vector, $Q$.
If any $c_i=0$,
then we remove the color $i$ from our alphabet of colors, and update
the value of $K$ accordingly.
The remainder of Bob's computation proceeds as a recursive divide-and-conquer
algorithm, which is similar in structure to the approach
of~\cite{c-m-83,g-oacmg-09}.

The generic problem is to determine the offset values of all the elements in
a range $Q[l..r]$, which initially is the entire vector
$Q=Q[0..N-1]$, assuming we know the values of 
$c_0,c_1,\ldots,c_{K-1}$, of every
color in $Q[l..r]$, and each $c_i>0$.
If $K\le 1$, we are done; so let us assume without loss of generality
that $K\ge 2$.
In addition, we assume inductively that we know, $d$, the number of
instances of color $0$ outside of the range $Q[l..r]$. Initially, of
course, $d=0$.

Given this initial setup, we split
$Q[l..r]$ into $Q[l..m]$ and $Q[m+1..r]$, where $m$ is in the middle
of the interval $[l,r]$. 
The main challenge, then, is to provide for
$Q[l..m]$ and $Q[m+1..r]$
the same setup we had for $Q[l..r]$.
This setup can be accomplished by determining the cardinalities,
$x_0,x_1,\ldots,x_{K-1}$ and 
$y_0,y_1,\ldots,y_{K-1}$,
of every color that respectively appears in 
$Q[l..m]$ and $Q[m+1..r]$.
We do this with a series of $K-1$ additional queries, where we guess
that the elements in $Q[l..m]$ are of color $i$, for $i=1,2,\ldots,K-1$, and
that the rest of $Q$ is of color $0$.
Let the values of these queries be denoted as $b_1,b_2,\ldots,b_{K-1}$,
and note that, at this point, we know the following:
\begin{eqnarray}
\label{eq-1}
x_i + y_i = c_i, \mbox{\ \ \ for $i=0,1,\ldots, {K-1}$} \\
\label{eq-2}
x_i + y_1 = b_i-d, \mbox{\ \ \ for $i=1,2,\ldots,K-1$} \\
\label{eq-3}
x_0+x_1+\cdots+x_{K-1}= m-l+1.
\end{eqnarray}
Thus, we can determine $y_0$, as
\[
y_0 = \frac{c_0 + \sum_{i=1}^{K-1} (b_i-d) - (m-l+1)}{K} ,
\]
for $y_0$ is counted $K$ times in the sum of $c_0$ and all the
$(b_i-d)$'s, and the sum of the $x_i$'s is $m-l+1$, by
Equation~(\ref{eq-3}).
Given the value of $y_0$, we can then determine all the $x_i$ values,
by using Equation~(\ref{eq-1}) for $x_0$ and 
Equation~(\ref{eq-2}) for $x_1,x_2,\ldots,x_{K-1}$.
Moreover, once we have all these $x_i$ values, we can determine the
values, $y_1,y_2,\ldots, y_{K-1}$, using Equation~(\ref{eq-1}).
Finally, we can determine the values $d'=d+y_0$ and 
$d''=d_{x_0}$ and use these respectively for the role of $d$ in
$Q[l..m]$ and $Q[m+1..r]$.
This gives us all the values necessary to then recursively determine
$Q[l..m]$ and $Q[m+1..r]$.
Of course, if the $c_i$ values for either of these subproblems are all $0$,
except for one (which would be equal to the size of this problem), then there
is no need to recursively solve this problem; so we would not perform a
recursive call in this case.

Let us, therefore, analyze the number of vector guesses performed by
this algorithm.
Ignoring for the time being the initial set of $K$ guesses,
note that we only
continue to search if we are guaranteed to be honing in on a substitution.
Thus, adding back the initial $K$ guesses,
we get that the total number of guesses is at most
\[
s(Q)\lceil \log M\rceil + K.
\]
Thus, we have the following.

\begin{theorem}
\label{thm:deterministic2}
Given an unknown length-$N$ sequence $Q$, defined on an alphabet of
size $K$, with $Q$ having $M$ possible locations of deviation from a
reference sequence, $R$, a malicious Mastermind attacker can discover 
$Q$ in polynomial time using 
$s(Q)\lceil\log M\rceil+K$ guesses, 
each of which reveals only the number of positions
where $Q$ and the test sequence match and
where $s(Q)$ denotes the number of substitutions that 
would transform $R$ into $Q$.
\end{theorem}

As we note in Section~\ref{sec:experiments}, this performance is more
than adequate to show that nearly half of all
mitochondrial DNA data in our case study
are vulnerable to this version of the
Mastermind attack.
Before we provide those statistics, however, let us study how the
Mastermind attack with sequence-alignment queries can be streamlined
to exploit DNA data distributions.

\subsection{The Sequence-Alignment Case}
\label{sec:align-R}
As mentioned above, roughly half
of the sequences in the mitochondrial DNA data set include
insertions and/or deletions in addition to substitutions in the
reference sequence, $R$.
Thus, we discuss in this subsection how we can modify the
Mastermind attack algorithm of Section~\ref{sec:align} to take
advantage of the distributional properties common in genomic data
sets, so as to discover a query sequence that can have arbitrary
kinds of differences with the reference sequence, $R$.
In this case, we view differences with $R$ procedurally 
as \emph{events}, each of
which is either a singleton deletion,
or an arbitrary-length insertion, which would transform $R$ into the
query sequence, $Q$.
(Note: for this algorithm, we view a substitution as actually occurring as a
deletion event followed by an insertion event.)

In this case, we run the attack algorithm in two phases. In Phase~1,
we aim to discover all the deletion events, and in
Phase~2, we aim to discover all the insertion events.
In both phases, we make the simplifying assumption that 
insertion and deletion events are disjoint. That is, they don't
overlap or interfere with one another. This assumption is based on the fact that
these events come from a statistical characterization of genomic
sequences, which is designed to
keep events disjoint (for overlapping events are better subdivided
further and considered as separate sub-events).
So, for example, we assume that there is no insertion event that is
then followed by a deletion event that then removes part of the
sequence that was just inserted. 

We begin by performing a guess for the reference sequence, $R$.
Armed with the sequence-alignment score, $a(Q,R)$, for $R$,
we then perform a divide-and-conquer
computation to find all the deletion 
events that occur in going from $R$ to $Q$.
Note that if we next perform a guess $V$ for a collection of deletion events
at some subset of
the $M$ statistically possible (deletion) locations in $R$, then we can
detect how many deletions actually occurred at these locations.
Moreover, note that the insertion events don't change this score, since the
insertions and deletions do not interfere, by assumption.
For each deletion event that is present in one of the queried locations, 
then our score
will not change with respect to the score for $R$, and,
for each location that should not be deleted, we will record a
score for $V$ that is one worse than that for $R$.
Thus, we can determine the number of deletion events for any test we
do by the difference between the score we observe and 
the score we would expect if all of the deletions are removing 
actual matches.
That is, if we test for $r$ singleton deletion events 
in $V$, then the number that actually occur is
$a(Q,V)-(a(Q,R)-r)$, where $a$ is the sequence-alignment score
function.

Let $Z_{1,M}=\{z_1,z_2,\ldots,z_M\}$ be a set of Boolean variables, such that
$z_i$ is 1 if and only if the $i$th statistically possible deletion event in
$R$ actually occurs in going from $R$ to $Q$.
We can perform a divide-and-conquer search in $Z_{1,M}$ to 
determine which of the
$z_i$'s are 1.
We begin by testing for all the deletion events in $Z_{1,M}$.
This gives us the number of $1$'s in $Z_{1,M}$.
We then perform a test for every deletion event in 
$Z_{1,M/2}=\{z_1,\ldots,z_{M/2}\}$, which by deduction gives us the number in
$Z_{M/2+1,M}=\{z_{M/2+1},\ldots,z_{M}\}$.
We then recursively determine the number in either or both of these two sets so 
long as there is at least one deletion event in that set.
Thus, we perform a divide-and-conquer parallel ``binary'' search
for each of the exact locations of singleton deletions.
Once we have completed this computation for $R$, with queries against $Q$,
we will have
determined the locations of all the deletion events from
$R$ to $Q$, including those deletions that are really substitution events.
Thus, this set of guesses 
uses at most $1+d(Q)\lceil\log M\rceil$
tests, where $d(Q)$ is the set of (singleton) deletion events in
going from $R$ to $Q$.

Once we know the locations of all the deletions in going from $R$ to $Q$,
we perform a second set of binary searches, just among these
locations, to find the locations among this group 
that are actually the sites of substitution events.
Let us now define $R'$ to be the reference sequence resulting from 
performing the events we discovered in Phase~1.
In particular, we perform a binary search for each of the $K$ colors,
with respect to $R'$,
searching, for each color $i$, in the statistically
possible insertion locations in $R'$
where we improve our score by adding a single character
of color $i$.
Note that there may be more than a single character of color $i$ inserted at
this location, but it is sufficient to do a single character query to
determine that there is an insertion here, since there is a non-deleted
element between every possible insertion location in $R'$.

Since we continue to perform recursive binary-type searches 
for any insertion locations that actually cause insertions, then the 
the set of additional guesses we do in this part of the second phase is 
at most $K + e(Q)\lceil\log M\rceil$,
where $e(Q)$ is the number of insertion events.

At this point 
in the algorithm, we know where all the insertion events are located, but we
don't know the full extent of each of their sizes.
So for each location, we perform a set of $K$ guesses of length $2$ to see if
we get a higher score by considering a longer insertion.
If there are no differences from the singleton queries, then we
can infer the length of the insertion from the previous queries.
Otherwise, we perform a set of $K$ guesses of length $3$, $4$, and so on,
until we observe no change from the previous set of guesses.
Thus, with a total number of guesses equal to $K\varepsilon(Q)$, where
$\varepsilon(Q)$ is the total size of all the insertion events, we discover
the length of each insertion event.
To complete the computation, then, we perform a miniature version of
our algorithm from Section~\ref{sec:align} at each
location determined to be to site of an insertion event.
Each such computation requires $(m+1)K$ guesses, where
$m$ is the length of the insertion.
Thus, the total number of guesses made in this part of Phase~2 is
$(\varepsilon(Q)+1)K$.
Therefore, we have the following.

\begin{theorem}
\label{thm:align2}
Given an unknown length-$N$ sequence $Q$, defined on an alphabet of
size $K$, with $Q$ having $M$ possible locations of deviation from a
reference sequence, $R$, a malicious Mastermind attacker can discover 
$Q$ in polynomial time using 
$(d(Q)+e(Q))\lceil\log M\rceil + (\varepsilon(Q)+2)K + 1$
guesses, each of which reveals only the number of positions
where $Q$ and the test sequence match,
using sequence-alignment LCS tests, where
\begin{itemize}
\item
$d(Q)$ is the number of deletion events,
\item
$e(Q)$ is the number of insertion events,
\item
$\varepsilon(Q)$ is the total length of all insertion events.
\end{itemize}
\end{theorem}

\section{Case Study for Mitochondrial DNA}
\label{sec:experiments}
We are at the point where hundreds of
thousands of people have had
their mitochondrial DNA (mtDNA) sequenced~\cite{ps-mdhe-05,brb-gppd-05},
which is typically about 16,500 base pairs (bp) long, 
whereas the entire diploid human genome is roughly 6 billion bp long.
Interestingly, since mtDNA is transferred only along the maternal line,
scientists have used differences from a reference mtDNA sequence
as a way to plot human migration from the earliest days of
the modern human species.
(See Figure~\ref{fig:migration}.)

\begin{figure}[hbt]
\begin{center}
\includegraphics[width=3.1in]{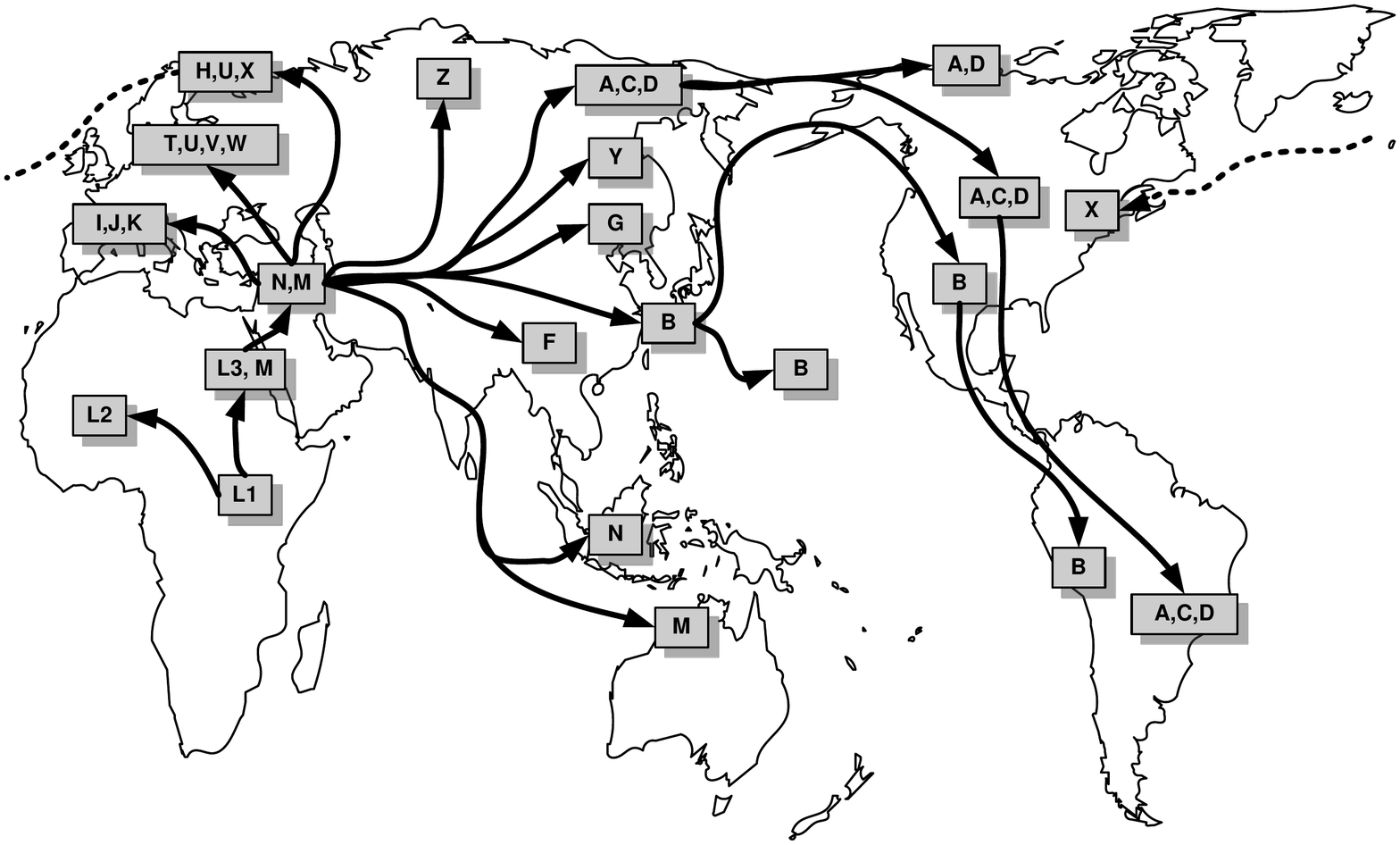}
\vspace*{-26pt}
\caption{
A confluent illustration~\cite{degm-cdvnp-03} of
the pattern of human migration implied by 
mtDNA mutations~\cite{ps-mdhe-05,brb-gppd-05}.
Each letter stands for a major human mitochondrial haplogroup, that
is, a canonical set of genetic mutations from a common ancestor.
}
\label{fig:migration}
\end{center}
\end{figure}

Because of this knowledge of migration patterns and its correlation to 
known mtDNA mutations, given someone's mtDNA sequence, 
it is possible to trace their maternal ancestry back to individual
villages~\cite{brb-gppd-05}, 
just by identifying differences in their mtDNA to a reference sequence,
e.g., rCRS (see Figure~\ref{fig:rCRS}).
In other words, mtDNA alone is sufficient to determine 
a person's ethnic background
with incredible accuracy.
Thus, we are at a point where privacy is a real concern with
respect to genomic sequences, and this concern is sure to increase in
the future.

\begin{figure}[hbt]
\medskip
\begin{center}
\begin{minipage}{2.25in}
\begin{verbatim}
GATCACAGGTCTATCACCCTATTAA
CCACTCACGGGAGCTCTCCATGCAT
TTGGTATTTTCGTCTGGGGGGTATG
CACGCGATAGCATTGCGAGACGCTG
GAGCCGGAGCACCCTATGTCGCAGT
ATCTGTCTTTGATTCCTGCCTCATC
...
ATCTGGTTCCTACTTCAGGGTCATA
AAGCCTAAATAGCCCACACGTTCCC
CTTAAATAAGACATCACGATG
\end{verbatim}
\end{minipage}
\end{center}
\caption{
A portion of the Revised Cambridge Reference Sequence,
rCRS (GenBank accession number: AC 000021), which
is 16,568 bp long.
}
\label{fig:rCRS}
\end{figure}

In addition to ethnicity, there are, of course, other privacy
concerns with respect to genomic data, including sensitive
information related to disease susceptibility, and possible genetic
influences on sexual orientation, personality, addiction, and intelligence.
Concerns that employers or insurers will use genetic information to screen 
those at high risk for a disease are already a public concern
and stories involving such risks are widespread in the press.
Indeed, the U.S.~government and 
several states have already created laws dealing with DNA data access, 
and many more are considering such legislation.
Thus, there is a need for technologies that can safeguard the privacy 
and security of genomic data.

Fortunately, several researchers have started
exploring privacy-preserving data querying
methods that can be applied to
genomic sequences (e.g., see~\cite{akd-spsc-03,da-smc-01,FNP04}).
That is, cryptographic techniques 
can be used to allow for queries to be 
performed in a way that answers the specific
question---such as a score rating the quality of a query
for DNA matching or sequence alignment---but does not
reveal any other information about the data, such as race or disease risk of
the individual whose DNA is being queried.

The purpose of this case study is to show that, while being sufficient for 
single-shot comparisons of DNA sequences,
such cryptographic techniques have a weakness
when they are employed repeatedly.
Specifically, we explore in this section 
how the Mastermind attack
allows a genomic querier, Bob, to iteratively discover 
the full identity of a genomic
query sequence, $Q$, with surprising efficiency, 
even if each comparison of $Q$ with
Bob's sequences are done using cryptographic
privacy-preserving protocols.
It is not surprising that iterated privacy-preserving sequence comparisons leak
some information about the sequences being compared; what is surprising is how 
quickly the Mastermind attack can work, especially on genomic data.

To demonstrate the vulnerability of real-world DNA data
to the Mastermind attack, we have performed
a case study of our distribution-based Mastermind attack
algorithms.
We used 1000 human mitochondrial
sequences downloaded from a recent version of GenBank
(http://www.ncbi.nlm.nih.gov/Genbank/index.html).
We focused on the sequences alone, ignoring any header and
other information, and have simulated Mastermind attacks on each one. 
The Revised Cambridge Reference Sequence
(rCRS) (GenBank accession number: AC 000021) was also downloaded
and used as the reference sequence~\cite{brandon08,ruiz07,brandon04}. 
The reference sequence is 16,568 bp long. All the sequences were aligned to 
the reference sequence and, for each sequence, the indices of the location
of each variation
were recorded together with the type (substitution, insertion, deletion) and
content of each variation.
This step is also essential if one is interested in compressing the 
data~\cite{baldicompression07}, for example.
Statistics for the number of substitutions, deletions, and insertions for 
this data
set of 1000 mtDNA sequences is given in Table~\ref{tbl:stat1}.

\begin{table}[hbt]
\begin{center}
\begin{tabular}{l|cc}
& \textbf{mean} & \textbf{standard dev.} \\
\hline
\rule{0pt}{12pt}\textbf{Substitutions} & 28.00 & 18.38 \\
\textbf{Deletions} & 0.90 & 2.46 \\
\textbf{Insertions} & 0.95 & 1.10 \\
\hline
\end{tabular}
\end{center}

\caption{Frequency statistics for 1000 mtDNA sequences. Mean and standard
deviation statistics are given for the frequency of substitutions, deletions,
and insertions in going from 
the reference sequence, $R=\mbox{rCRS}$,
to each sampled sequence.}
\label{tbl:stat1}
\end{table}

Of the 1000 sequences, 453 
have only substitution events with respect to the reference sequence,
$R=\mbox{rCRS}$.
So we used this subset of 453 sequences to test the 
simulated performance of the method
of Theorem~\ref{thm:deterministic2}.
The distribution of the number of substitutions in each of these sequences is
shown in Figure~\ref{fig:subs}.

\begin{figure}[hbt]
\begin{center}
\includegraphics[width=3in]{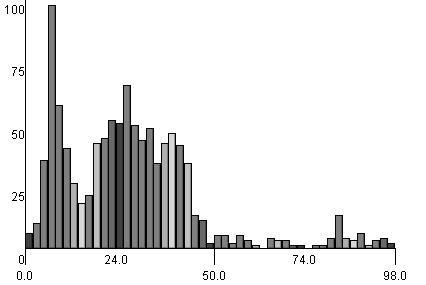}
\caption{Histogram of number of substitutions in 
1000 mtDNA with respect to the reference sequence, $R=\mbox{rCRS}$.}
\label{fig:subs}
\end{center}
\end{figure}

Note that these frequencies do not follow a normal distribution, which shows
the importance of our using real-world data, such as this, rather than
randomly-generated or simulated data.
The statistical diversity of the mtDNA data is actually a 
reflection of the racial diversity of the people whose mtDNA data is included
in our data set.
That is,
edit distance from the reference sequence, $R=\mbox{rCRS}$, 
across the human species, 
is not uniformly or normally distributed.
Instead, edit distance from rCRS is a reflection of human migration patterns,
as illustrated in Figure~\ref{fig:migration}.

The 45.3\% of the sampled mtDNA sequences with substitution-only modifications 
from rCRS
are exactly the set of sequences that can be effectively discovered by the
single-count Mastermind attack of Theorem~\ref{thm:deterministic2}.
Thus, we simulated the performance of this attack on each one of these
sequences and tabulated the number of guesses that would be needed in each
case in order to discover the complete identity of each sequence.
Interestingly, 90\% of the simulated
substitution-only Mastermind attacks completed with 375
guesses or less.
The complete distribution of single-count
Mastermind attack lengths for this data set are
shown
in Figure~\ref{fig:subonly}.

\begin{figure}[hbt]
\begin{center}
\includegraphics[width=3in]{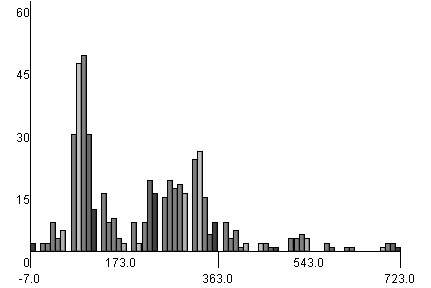}
\caption{Histogram of Mastermind attack lengths
for 453 substitution-only mtDNA sequences with 
standard single-count Mastermind scores.
The mean attack length for this data set was
219.6 and the standard deviation was 139.1.
}
\label{fig:subonly}
\end{center}
\end{figure}

All 1000 sampled mtDNA 
sequences were then used to test the performance of the method of
Theorem~\ref{thm:align2}.
Sequence-alignment Mastermind attacks were
simulated for each such mtDNA sequence while the number of sequence-alignment
tests were counted for each.
Interestingly, 90\% of these simulated
subsequence-alignment Mastermind attacks completed with 875 
guesses or less.
And some completed with much fewer than this.
The complete distribution of sequence-alignment Mastermind attack lengths for
this data set is shown in Figure~\ref{fig:subseq}.

\begin{figure}[hbt]
\begin{center}
\includegraphics[width=3in]{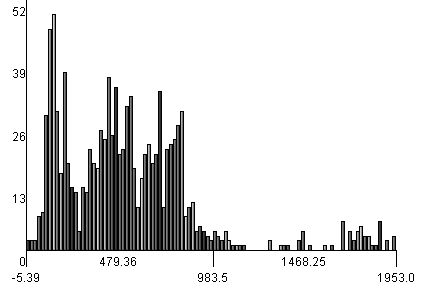}
\caption{Histogram of simulated Mastermind attack lengths
for 1000 mtDNA sequences with 
sequence-alignment scores.
The mean sequence-alignment simulated Mastermind attack length was
536.3 with a standard deviation of 373.9.
}
\label{fig:subseq}
\end{center}
\end{figure}

\section{Discussion and Future Directions}
We have shown that, even though the single-count and sequence-alignment 
Mastermind satisfiability problems are NP-complete, 
one can effectively mount Mastermind attacks on arbitrary genomic sequences
just by knowing basic information about the length of the sequences and the
number of characters in the alphabet used to construct those sequences.
Moreover, if one has some basic
statistical information about these sequences, relative to a reference sequence,
then one can mount the Mastermind attack with surprising effectiveness.
In fact, we provided a case study suggesting
that such attacks are already possible
and surprisingly efficient for mtDNA sequences.

One conclusion to draw from this work
is that privacy-preserving protocols for performing
a query with a sequence, $Q$,
against a genomic database, $D$,
should take into account the entire 
set of comparisons~\cite{da-psrda-01}, 
with $Q$ and the sequences in $D$, rather than relying on the
privacy-preservation of each individual comparison in turn.
For example, in the usage model where Bob is a user querying a genomic 
database, the Mastermind attack is weakened if it is difficult
for Bob to know the index of the sequences he is comparing against---for
example, if the database owner, Alice, presents her sequences in a different
random order each time. Such an obfuscation does not defeat the Mastermind
attack, however, if Bob is able to use other reasoning inferences to match
scores of his query sequences across multiple queries in Alice's database
of sequences.

In terms of further exploration of the vulnerability of genomic data to the
Mastermind attack, one interesting direction for future work would be to
test the vulnerability of entire human genomes to the Mastermind attack, once
we have enough completed genomes to do such an experimental study.
In addition, other
directions for future research therefore could include new, efficient
privacy-preserving schemes for querying entire genomic databases with
respect to sequence-alignment queries.
Such results would negate the privacy-exposing vulnerabilities of the
Mastermind attack.

\subsection*{Acknowledgments}
We would like to thank Pierre Baldi 
for suggesting the security of genomic data as an important research question
and for providing the mitochondrial DNA data used in our experiments, including
the characterizations in terms of the reference sequence, rCRS.
We would also like to thank 
David Eppstein,
Daniel Hirschberg,
Stas Jarecki,
and
Michael Nelson
for helpful discussions regarding the topics of this paper.
This research was supported in part by the National Science
Foundation under grants 0724806, 0713046, and 0847968.
Some of the results of
this paper appeared in preliminary
form as~\cite{g-magd-09}, albeit with some flawed arguments for justifying
previous versions of Theorems~\ref{thm:seq-align-q}
and~\ref{thm:align2}.

\ifArxiv
{\raggedright
\bibliographystyle{abbrv}
\bibliography{goodrich,baldi,info,ref,genome,math,security}
}
\fi

\end{document}